\documentclass[twocolumn,article]{IEEEtran}

\usepackage[english]{babel}
\usepackage[latin1]{inputenc}
\usepackage{amssymb,amstext,amsmath}
\usepackage[ruled,vlined]{algorithm2e}
\usepackage{verbatim}

\usepackage{float}

\usepackage{cite}

\ifCLASSINFOpdf
   \usepackage[pdftex]{graphicx}
   \DeclareGraphicsExtensions{.pdf,.jpeg,.png}
\else
   \usepackage[dvips]{graphicx}
   \DeclareGraphicsExtensions{.eps}
\fi

\newtheorem{theorem}{Theorem}

\begin{document}

\title{Secure Data Transmission in Cooperative Modes: Relay and MAC}


\author{\IEEEauthorblockN{Samah A. M. Ghanem\IEEEauthorrefmark{1},
Munnujahan Ara\IEEEauthorrefmark{2},\\
\IEEEauthorblockA{{\IEEEauthorrefmark{1}}Mobile Communications Department,\\
Eurecom Institute, Campus SophiaTech, 06410 Biot, France\\
Email: samah.ghanem@eurecom.fr;}\\
\vspace{0.2cm}
\IEEEauthorblockA{{\IEEEauthorrefmark{2}}Mathematics Department,\\
Khulna University, Khulna-9208, Bangladesh\\
Email: munnujahan@gmail.com}}}

\maketitle

\begin{abstract}
\boldmath
Cooperation in clouds provides a promising technique for 5G wireless networks, supporting higher data rates. Security of data transmission over wireless clouds could put constraints on devices; whether to cooperate or not. Therefore, our aim is to provide analytical framework for the security on the physical layer of such setup and to define the constraints embodied with cooperation in small size wireless clouds. In this paper, two legitimate transmitters Alice and John cooperate to increase the reliable transmission rate received by their common legitimate receiver Bob, where one eavesdropper, Eve exists. We provide the achievable secure data transmission rates with cooperative relaying and when no cooperation exists creating a Multiple Access Channel (MAC). The paper considers the analysis of different cooperative scenarios: a cooperative scenario with two relaying devices, a cooperative scenario without relaying, a non-cooperative scenario, and cooperation from one side. We derive analytical expressions for the optimal power allocation that maximizes the achievable secrecy rates for the different set of scenarios where the implication of cooperation on the achievable secrecy rates was analyzed. We propose a distributed algorithm that allows the devices to select whether to cooperate or not and to choose their optimal power allocation based on the cooperation framework selected. Moreover, we defined distance constraints to enforce the benefits of cooperation between devices in a wireless cloud.
\end{abstract}

\vspace{0.3cm}


\IEEEpeerreviewmaketitle
%

\section{Introduction}
%
The Wiretap channel models scenarios of the data transmission under security attacks on the physical layer~\cite{Wyner1975}. This paper focuses on the physical layer security in a wiretap channel model of interest; including two legitimate transmitters that are relay devices, one legitimate receiver, and one eavesdropper. Several optimal power allocation interpretations that aim to maximize the secure and reliable information rates are existing in the literature. Such designs were done for different channel models, for example, for the two user MAC Gaussian channel~\cite{SamMAC12}, or for cooperative virtual MIMOs~\cite{SamMul13} by directly maximizing the mutual information, or via optimizing other design criterion such as, minimizing the mean square error~\cite{Arahugo11}, or minimizing the bit error rate~\cite{BarRod2006}. In~\cite{DongPoor}, the authors address secure communications of one source-destination pair with the help of multiple cooperating relays in the presence of one or more eavesdroppers with different cooperative schemes. In \cite{RMYS06}, the authors provide upper bounds of the achievable rates of a discrete memoryless MAC channel with confidential messages are to be transmitted in perfect secrecy. In~\cite{HElGamal}, the authors devise several cooperation strategies and characterize the corresponding achievable rate-equivocation region. They consider a deaf helper phenomenon, where the relay is able to facilitate secure communications while being totally ignorant of the transmitted messages. In~\cite{Ooahama}, the author studied the security of communication for the relay channel under the situation that some of the transmitted messages are confidential to the relay. Moreover, in~\cite{DongPoor2}, the authors considered cooperative jamming where a relay equipped with multiple antennas transmits a jamming signal to create interference at the eavesdropper. They proposed design methods to determine the antenna weights and transmit power of source and relay, so that the system secrecy rate is maximized. 

\vspace{0.3cm}

In this paper, we consider a scenario which is more of practical relevance where two side cooperation exists. The usual assumption of one side cooperation is addressed for analytical purposes only. However, we consider cooperation of bi-directional two relay devices to show that 'a real egoistic behavior is to cooperate',~\cite{FF}. Our work differs from other works not only on this assumption, but we have also considered that the devices under certain distance constraints switch their mode of cooperation from Relay to MAC or vice versa or choose not to cooperate. This assumption is of particular relevance in a device to device cooperation within a wireless cloud to assure that cooperation will not harm one or the other device reliable and secure transmission rates. 

\vspace{0.3cm}

Recently, another branch in the field of physical layer security relied on game theory. The maximization of reliable information rates with different channel models with cooperative relaying and with the existence of a jammer are considered in~\cite{IFS08}. Optimal power allocation strategies were derived for a zero sum game with an unfriendly jammer in~\cite{Ara2012}. In~\cite{MXE2010}, the authors studied the achievable secrecy rates for a non-cooperative zero sum game with a jamming relay assisting the eavesdropper. They found a saddle point solution of such game. In~\cite{MMH2010}, the authors propose a distributed game-theoretic method for power allocation in bi-directional cooperative communication. They showed that their method reaches equilibrium in one stage, and proved the benifits of bi-directional cooperation between nodes closer to each other. In~\cite{GLOBCOM08}, the authors address the power allocation problem for interference relay channels. They model the problem as a strategic non-cooperative game and show that this game always has a unique Nash equilibrium. In~\cite{girlfriend}, the authors investigate the interaction between the source and friendly jammers, they introduce a game theoretic approach in order to obtain a distributed solution. Moreover, in~\cite{MHM12}, the authors investigate optimal power allocation strategies for OFDM wiretap channels with the existence of friendly jammers. 

\vspace{0.3cm}

In this paper, the wiretap channel model of interest has no jammers, however, bi-directional cooperative communication between relay devices are considered. Here, we are particularly concerned about optimal cooperative power allocation strategies which allow maximum reliable and achievable secrecy rates. Despite the fact that our problem can be formulated using game theory. However, the focus here is much more in finding strategies in a distributed sense which can allow the devices to choose first their cooperation mode, then to allocate their power control strategy accordingly.

\vspace{0.3cm}

The broadcast nature of a wireless medium allows multiple devices to transmit and receive simultaneously. Such nature, allows some devices to choose to cooperate in fixed or mobile clouds for their shared benefits, or to overhear and target multiple transmitting devices through their direct transmission, or over their relayed transmission. In this paper, we focus on a similar setup, where transmitting devices cooperate, while an eavesdropper overhear their own transmissions, assuming that this eavesdropper is only overhearing their own direct transmissions. In fact, this assumption is basically based on the lack of knowledge of the eavesdropper - who aims to decode their transmitted messages - that a message of one transmitter could be mixed over time or that any cooperation could exist. This assumption could also simplify the mathematical setup of the problem. Of particular relevance are the benefits of cooperation to secure data transmission, and more relevant is to study when and where cooperation should exist, building a framework of distance constraints which could allow devices in a cloud to decide to go for cooperation, to cooperate from one side, not to cooperate, or to change location avoiding any distance attacks. We mean by a distance attack, is the capability of one eavesdropper device to experience a better version of the transmitted message than the legitimate receiving device. Under such distance constraints, the legitimate transmitters and receivers could choose to move far from an attacking device as a defense strategy.

\vspace{0.3cm}

In this paper, we build a framework for the achievable secrecy rates - defined by upper bounds for relay devices - capitalizing on the achievable rates for MAC channels. Then we derive the optimal power allocation strategies under different scenarios, where two legitimate transmitters/receivers Alice and John cooperate in a bi-directional way to increase their secrecy rate, i.e., to increase the secure and reliable information transmission rate received by Bob, their common legitimate receiver during which an eavesdropper, Eve tries to eavesdrop both. Four scenarios have been analyzed, a cooperative scenario where relaying is used between Alice and John, cooperation without relaying, a non-cooperative scenario, and a scenario with cooperation from one side. In the four scenarios, the optimal power allocation for both legitimate transmitters has been derived to maximize their secure achievable rates against Eve. The solution set is the optimal response from Alice and John. We finally consider the distance of Eve as the limiting constraint that defines how much the cooperation will be of benefit.

\vspace{0.3cm}

The paper is organized as follows, section II explains the model used throughout the paper, section III introduces the model achievable secrecy rate regions. In section IV, the scenarios, the implications on their achievable secrecy rates, the formulation of the optimization problem, and the optimal power allocation are introduced. Section V builds upon distance constraints for the model under study. Section VI presents the proposed algorithm, and we conclude the paper with numerical and analytical results.
%

\section{SYSTEM MODEL}
\label{Problem}
%

We consider a model that includes two legitimate transmitters Alice and John, and one common legitimate receiver, Bob. One eavesdropper Eve tries to decrease the security of both transmitters trying to decode the messages received by both transmitters. The communications between different parties are done over a point to point bi-directional links. Figure \ref{fig:figure-1} illustrates the system model\footnote{Notice that the model considers that the power $P_A$, $P_J$ for Alice and John, respectively is divided between the main transmitted signals and the relayed ones.}.
\begin{figure}[ht!]
    \begin{center}
      \mbox{\includegraphics[height=3in,width=3.5in]{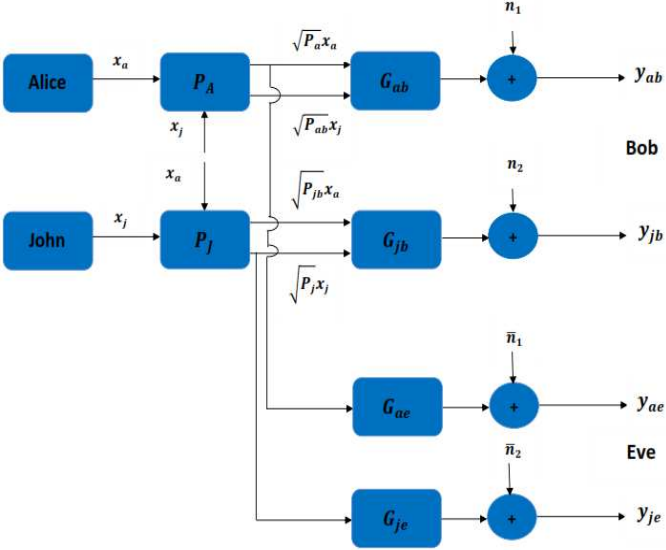}}
          \caption{The system model.}
    \label{fig:figure-1}
    \end{center}
    \label{figure-1}
\end{figure}

\vspace{0.3cm}

The transmitted message from transmitter $i$ is defined as $x_{i}$, and the received message by the receiver $k$ is defined as $y_{k}$. Considering that Alice and John are relay devices who cooperate in relaying each others data to Bob who receives two vectors from Alice and from John, assuming that each will relay a replica of the other's main message as follows,

\begin{equation}
y_{ab}= G_{ab}\sqrt{P_a} x_{a}+G_{ab}\sqrt{P_{ab}} x_{j} + n_{1}.
\end{equation}
\begin{equation}
y_{jb}= G_{jb}\sqrt{P_j} x_{j}+G_{jb}\sqrt{P_{jb}} x_{a} + n_{2}.
\end{equation}

However, Eve receives two vectors from Alice and John; assuming that Eve will receive the main message of each via overhearing, and will not be aware of the relayed part, this assumption is done for the sake of simplicity, as follows,

\begin{equation}
y_{ae}= G_{ae}\sqrt{P_a} x_{a}+ \bar{n}_{1}.
\end{equation}
\begin{equation}
y_{je}= G_{je}\sqrt{P_j} x_{j}+ \bar{n}_{2}.
\end{equation}

$y_{ab}\in \mathbb{C}^n$ and $y_{jb}\in \mathbb{C}^n$ represent the received vectors of complex symbols at Bob's side from Alice and John; respectively, $y_{ae}\in \mathbb{C}^n$ and $y_{je}\in \mathbb{C}^n$ represent the received vectors of complex symbols at Eve's side from Alice and John; respectively. $x_{a}\in \mathbb{C}^n$ and $x_{j}\in \mathbb{C}^n$ represent the vectors of complex transmit symbols with zero mean and identity covariance $\mathbb{E}[x_{a}x_{a}^{\dag}], \mathbb{E}[x_{j}x_{j}^{\dag}]$, respectively. $n_{1}\in \mathbb{C}^n$, $n_{2}\in \mathbb{C}^n$, $\bar{n}_{1}\in \mathbb{C}^n$, and $\bar{n}_{2}\in \mathbb{C}^n$ represent vectors of circularly symmetric complex Gaussian noises with zero mean and identity covariance. $G_{ik}$ represent the complex gains\footnote{Notice that the channel gains are considered to be fixed over the transmission time of each symbol $(x_a,x_j)$. This means that the transmission time is contained in the interval of the coherence time of the channel between each legitimate transmitter and the legitimate receiver. Therefore, $G_{ab} (\sqrt{P_a}x_a+\sqrt{P_{ab}}x_j)$ and $G_{jb} (\sqrt{P_{jb}}x_a+\sqrt{P_{j}}x_j)$ are associated to the transmission of the main and relayed transmitted symbols over each link.} of the channels between transmitter $i$ and receiver $k$. $\sqrt{P_{jb}}$ and $\sqrt{P_{ab}}$ represent the relay power used by John and Alice  respectively to relay each others data. $\sqrt{P_{a}}$ and $\sqrt{P_{j}}$ represent the transmitted power for Alice and John used respectively to transmit their own data.

\vspace{0.3cm}

If  relaying is precluded, the second sum term including $\sqrt{P_{jb}}$ and $\sqrt{P_{ab}}$ will be omitted from the equations. The achievable secrecy rate that each of the legitimate devices will try to maximize is called (Cs), assuming Maximum Ratio Combiner (MRC) at Bob's side.
%

\section{THE SECURE ACHIEVABLE RATE REGIONS}
\label{THE SECURE ACHIEVABLE RATE REGIONS}
%

\subsection{Achievability Framework}
%

We can consider that the achievable secrecy rate regions for the model in Figure 1 with two relay devices and one eavesdropper coincide with that of the MAC with one common eavesdropper, see~\cite{RMYS06},~\cite{Gamal2006},~\cite{gerhard11}. Therefore, we can conclude the following theorem.

\vspace{0.3cm}

\begin{theorem}
(Achievability): The achievable secure rate regions for the model in Figure 1 is upper bounded by the "secure" achievable rates for a MAC channel corresponding to each transmitter and their common receiver, and the two legitimate transmitters and their common eavesdropper. The rates in the closure of the union of all $(R1,R2)$ satisfy,

\begin{equation}
R1\leq I(x_{a};Y_{b}|x_{j})-I(x_{a};Y_{e}|x_{j}). 
\end{equation}
\begin{equation}
R2\leq I(x_{j};Y_{b}|x_{a})-I(x_{j};Y_{e}|x_{a}). 
\end{equation}
\begin{equation}
R1+R2\leq I(x_{a},x_{j};Y_{b})-I(x_{a},x_{j};Y_{e}). 
\end{equation}

\end{theorem}

\vspace{0.3cm}

\begin{proof}
We can re-write the rate regions considering some auxiliary random variables $U$, $V$, and a time sharing variable $Q$.

\begin{equation}
R1\leq I(x_{a};Y_{b}|x_{j},Q)-I(x_{a};Y_{e}|x_{j},Q). 
\end{equation}
\begin{equation}
R2\leq I(V,x_{j};Y_{b}|U,x_{a},Q)-I(V,x_{j};Y_{e}|U,x_{a},Q). 
\end{equation}
\begin{equation}
R1+R2\leq I(x_{a},x_{j},U,V;Y_{b}|Q)-I(x_{a},x_{j},U,V;Y_{e}|Q). 
\end{equation}

For some product distribution $U\rightarrow V\rightarrow (X_{a}, X_{j})\rightarrow y_{ab}, y_{jb}\rightarrow Y_{b}$, $y_{ae},~y{je}\rightarrow Y_{e}$. The joint probability distributions $p(u,v,x_{a},x_{j},y_{ab},y_{jb})$ factors as, $p(u)p(x_{a}|u)p(v)p(x_{j}|v)p(y_{ab}),y_{jb}|x_{a},x_{j}$ and the joint probability distributions $p(u,v,x_{a},x_{j},y_{ae},y_{je})$ factors as, $p(u)p(x_{a}|u)p(v)p(x_{j}|v)p(y_{ae}),p(y_{je}|x_{a},x_{j})$; and the time sharing variable $Q$ and through Markovity, $U\rightarrow V\rightarrow (X_{a},X_{j})\rightarrow (Y_{b},Y_{e})$, the claim is proved.
\end{proof}
%

\subsection{Problem Formulation}
\label{Problem Formulation}

The maximum achievable secrecy rate for Alice,
\begin{equation}
 Cs1 = max \ R_{ajb}-R_{ae}
\end{equation}
subject to the power constraint, 
\begin{equation}
 P_{j}+P_{jb}\leq P_{J}
\end{equation}
Therefore,
\begin{equation}
 {P_{jb}}^{*} = arg~\max_{\zeta P_{jb},P_{ab}/\zeta} R_{ajb}-R_{ae},
\end{equation}
with~${P_{j}}^{*}=P_{J}-{P_{jb}}^{*}$, and $0 \leq \zeta \leq 1$

\vspace{0.3cm}

$R_{ajb} = log (1+{SNR}_{ab}+{SNR}_{ajb})$, when John is relaying Alice data. When no relaying is considered $R_{ajb}$ is the same as,  

\begin{equation}
R_{ab} =I(x_{a};y_{ab})=log(1+{SNR}_{ab}).
\end{equation}
\begin{equation}
R_{ae} =I(x_{a};y_{ae})=log(1+{SNR}_{ae}).
\end{equation}

The maximum achievable secrecy rate for John,

\begin{equation}
 Cs2=max \ R_{jab}-R_{je}
\end{equation}
subject to the power constraint, 
\begin{equation}
P_{a}+P_{ab}\leq P_{A}
\end{equation}
Therefore,
\begin{equation}
 {P_{ab}}^{*} = arg~\max_{P{ab/\zeta,\zeta P_{jb}}} R_{jab}-R_{je},
\end{equation}
with~${P_{a}}^{*}=P_{A}-{P_{ab}}^{*}$, and $0 \leq \zeta \leq 1$

\vspace{0.3cm}

$R_{jab}=log (1+{SNR}_{jb}+{SNR}_{jab})$, when Alice is relaying John data. When no relaying is considered $R_{jab}$ is the the same as,
\begin{equation}
R_{jb}=I(x_{j};y_{jb})=log(1+{SNR}_{jb}).
\end{equation}
\begin{equation}
R_{je}=I(x_{j};y_{je})=log(1+{SNR}_{je}).
\end{equation}

Where $SNR_{ik}$ is the received signal to noise ratio between transmitter $i$ and receiver $k$.
\begin{equation}
{SNR}_{ik}=\frac{G_{ik}P_{i}}{\sigma^{2}}
\end{equation}	

$G_{ik}$ is the channel gain between different devices, and $\sigma^{2}$ is the noise power, considered as fixed over all links.

\vspace{0.3cm}

Given that the framework will include bi-directional cooperation, and including relaying with $\alpha$ cooperation level between Alice and John. The received $SNR$ via the path of transmitter $i$, relay point $r$, and receiver $k$ will be considered as follows,~\cite{MMH2010}
\begin{equation}
{SNR}_{irk}=\frac{G_{ir}G_{rk}P_{i}P_{rk}}{\sigma^{2}(G_{ir}P_{i}+ G_{rk}P{_{rk}}+\sigma^{2})}
\end{equation}

Therefore, we need to analyze a set of scenarios where we can derive the optimal power allocation required to maximize the achievable secrecy rates, and therefore to get insights on the effect of cooperation and relaying on the secrecy rates. In particular, we will devise the optimal power allocation set $({P_{ab}}^{*},{P_{jb}}^*)$ used for relaying, and $({P_{a}}^{*},{P_{j}}^{*})$ used for main data transmission. Where ${P_a}^{*}=P_A-{P_{ab}}^{*}$, and ${P_j}^*=P_J-{P_{jb}}^{*}$.

\section{COOPERATION FRAMEWORK}
\label{COOPERATION FRAMEWORK}

The cooperation setup in this paper is between two parties who cooperate to reach their optimum strategies in service request way, such that the one who request the relay service will follow a strategy of cooperation defined by the other device, where both devices at the end cooperate in relaying each others data, or not, in different cooperation levels. In particular, the mathematical formulation will show how cooperation is induced into the convex optimization problems. Based on the objective functions defined for each device, the devices choose the optimal power allocation, that may correspond to a cooperation decision when cooperation is of benefit or to minimal cooperation in a multiple access channel (MAC) mode, or no cooperation when there is no benefit expected.

\vspace{0.3cm}

Let's consider the objective functions per cooperative device, Alice with $Cs1$, and John with $Cs2$. However, distance considerations will be induced into the mathematical formulation to evaluate when the cooperation will be of benefit; this will be discussed later in section VIII in the paper.
%

\section{COOPERATIVE RELAYS}
\label{COOPERATIVE RELAYS}
%

Within the different scenarios considered, Eve is trying to eavesdrop both Alice and John. First, we consider the scenario when John is trying to relay Alice data with power $P_{jb}$ and Alice is trying to relay John data with power $P_{ab}$ using the shared bi-directional link between them (see Figure 1). Both relays utilize Amplify and Forward (AF) protocol for cooperation. The cooperation level defines the main cooperation point in the mathematical formulation of the optimization problem. Second, we consider the scenario where there is cooperation without relaying data. Third, we consider a scenario where no cooperation exists. Fourth, we consider a scenario when there is cooperation from one side and without relaying data. The later two scenarios considered give insightful solutions through which cooperation between devices in wireless clouds can be evaluated from a secrecy perspective. The optimization of each objective function is solved such that it is the solution of the derivative of the Lagrangian and applying the Karush-Kuhn-Tucker (KKT) conditions~\cite{Boyd04}. In particular, the achievable secrecy rate is maximized subject to fixed total power constraint for each transmitter. The solution of such optimization problem is the optimal power allocation that each device will admit.
%

\subsection{Cooperative Scenario with Relaying}
%

Consider the model in  Figure \ref{fig:figure-1}. To define the problem from a MAC setup perspective as introduced in Theorem 1. The achievable rate regions will follow the regions in Theorem 1. Therefore, the achievable secrecy rate and their corresponding optimization problem can be written as follows,
\begin{equation}
{Cs1}_{\alpha P_{jb}}=\max_{\alpha P_{jb}} \ log(1+{SNR}_{ab}+{SNR}_{jb})-log(1+{SNR}_{ae}).
\end{equation}
Subject to, 
\begin{equation}
P_{j}+P_{jb}\leq P_{J}
\end{equation}
\begin{equation}
{Cs2}_{\frac {P_{ab}}{\alpha}}=\max_{\frac {P_{ab}}{\alpha}} \ log(1+{SNR}_{ah}+{SNR}_{jb})-log(1+S{NR}_{je}).
\end{equation}
Subject to, 
\begin{equation}
P_{a}+P_{ab}\leq P_{A}
\end{equation}
%

However, such setup introduces a maximization over the upper bound of the achievable rates of the MAC at Bob without the upper bound on the eavesdropper MAC channel due to our previous assumptions. Therefore, this way the problem is not well defined, and we don't want to substitute the whole received SNR from Alice to Bob, or from John to Bob in both objective functions. Then, we will not be exactly optimizing over the MAC upper bound, but we will be optimizing over the relay capacity which induces the received SNR over the relay bi-directional link between Alice and Bob. It follows that, the achievable secrecy rate and their corresponding optimization problem will be as follows,
\small
\begin{equation}
\label{5}
{Cs1}_{\alpha P_{jb}}=\max_{\alpha P_{jb}} \ log(1+{SNR}_{ab}+{SNR}_{ajb})-log(1+{SNR}_{ae})
\end{equation}
\normalsize
Subject to, 
\begin{equation}
\label{5_sub}
P_{j}+P_{jb}\leq P_{J}
\end{equation}
\small
\begin{equation}
\label{6}
{Cs2}_{\frac{P_{ab}}{\alpha}}=\max_{\frac{P_{ab}}{\alpha}} \ log(1+{SNR}_{jb}+{SNR}_{jab})-log(1+{SNR}_{je})
\end{equation}
\normalsize
Subject to,
\begin{equation}
\label{6_sub}
P_{a}+P_{ab}\leq P_{A}
\end{equation}

\vspace{0.3cm}

Since the one who is providing the service first dictates the cooperation level, the formulation is defined as in~\eqref{5} and~\eqref{6}. Let Alice be the one who first request the relay service from John. Therefore, the power John decides to cooperate with will dictate the response of Alice. Hence, the first objective (utility) $Cs1$ will lead to the optimal relay power $P_{jb}$: John uses to relay Alice data so that he helps in increasing her utility, i.e., increasing her achievable secrecy rate. On the other hand, the second objective (utility) $Cs2$ is to maximize John achievable secrecy rate by letting Alice relay the data of John to Bob, thus we derive the optimal relay power $P_{ab}$: Alice uses to relay John data to Bob. $\lambda$ is a constant corresponding to the Lagrange multiplier. Thus, using the KKT conditions, see~\cite{Boyd04}, it follows an optimal power allocation policy. In particular, the point of optimality is decided by John, such that he optimizes his own objective, and because he wants to offer cooperation, he decides the level of cooperation, called $\alpha$ where, $0\leq \alpha<1$. Assuming that John will be more cooperative with Alice; the power cooperation is induced in the optimization problem as follows,

\begin{equation}
\label{7}
P_{ab}= \alpha P_{jb}      
\end{equation}
Solving~\eqref{5}, the cooperative optimal power allocation for the cooperative scenario through relaying which increases secrecy rate for Alice is as follows,

\begin{equation}
\label{8}
\psi_{1} {P_{jb}}^{3}+\psi_{2} {P_{jb}}^{2}+\psi_{3} P_{jb}+\psi_{4}=0
\end{equation}
where, \\
$\psi_{1}=\omega_{1}{\alpha}^{2}G_{ae}$; $\psi_{2}=\omega_{1}\omega_{4}+\omega_{2}{\alpha}^{2}G_{ae}$\\
$\psi_{3}=\omega_{2}\omega_{4}+\omega_{3}{\alpha}^{2}G_{ae}-\alpha G_{aj}G_{jb}G_{ae}P_{a}(G_{aj}P_{a}+{\sigma}^{2})$\\
$\psi_{4}=\omega_{3}\omega_{4}-\omega_{5}G_{aj}G_{jb}P_{a}(G_{aj}P_{a}+{\sigma}^{2})$\\ $\omega_{1}={\sigma}^{2}{G_{jb}}^{2}+{G_{jb}}^{2}G_{ab}P_{a}+{G_{jb}}^{2}G_{aj}P_{a}$\\
$\omega_{2}=({\sigma}^{2}G_{jb}+2G_{jb}G_{ab}P_{a}+G_{aj}G_{jb}P_{a}+{\sigma}^{2}G_{jb})(G_{aj}+{\sigma}^{2})$\\
$\omega_{3}=(G_{ab}P_{a}+{\sigma}^{2})(G_{aj}P_{a}+{\sigma}^{2})^{2}$\\
$\omega_{4}=\alpha{\sigma}^{2}+\alpha G_{ae}(1+\lambda P_{a})$; \ \ $\omega_{5}={\alpha}^{2}+G_{ae}\lambda P_{a}$

\vspace{0.3cm}

$P_{jb}^{*}$ is the solution of the third order equation, which is the optimal power allocation John will decide to cooperate with Alice to relay her data in order to increase her secrecy. 

\vspace{0.3cm}

Solving~\eqref{6}, the cooperative optimal power allocation for the cooperative scenario through relaying which increases secrecy rate for John is as follows,
\begin{equation}
\label{9}
\beta_{1} {P_{ab}}^{3}+\beta_{2} {P_{ab}}^{2}+\beta_{3} P_{ab}+\beta_{4}=0
\end{equation}
where,\\
\small
$\beta_{1}=G_{je}\phi_{1}$; $\beta_{2}=G_{je}\phi_{2}+\phi_{1}\phi_{4}$\\ $\beta_{3}=G_{je}\phi_{3}+\phi_{2}\phi_{4}-\alpha G_{ja}G_{ab}G_{je}P_{j}(G_{aj}P_{j}+{\sigma}^{2})$\\
$\beta_{4}=\phi_{3}\phi_{4}-\phi_{5}$\\
$\phi_{1}={\sigma}^{2}{G_{ab}}^{2}+{G_{ab}}^{2}G_{jb}P_{j}+{G_{ab}}^{2}G_{ja}P_{j}$\\
$\phi_{2}=({\sigma}^{2}G_{ab}+2G_{ab}G_{jb}P_{j}+G_{ab}G_{ja}P_{j}+{\sigma}^{2}G_{ab})(G_{ja}P_{j}+{\sigma}^{2})$\\
$\phi_{3}=(G_{jb}P_{j}+{\sigma}^{2})(G_{ja}P_{j}+{\sigma}^{2})^{2}$; \ \
$\phi_{4}=\alpha{\sigma}^{2}G_{je}\lambda P_{j}+\alpha G_{je}$\\ $\phi_{5}={\alpha}^{2}{\sigma}^{2}G_{ja}G_{ab}G_{je}\lambda P_{j}^{2}(G_{ja}P_{j}+{\sigma}^{2})$\\
\normalsize

$P_{ab}^{*}$ is the solution of the third order equation, which is the optimal power allocation Alice will use to cooperate with John to relay his data in order to increase his secrecy. Therefore, the optimal power allocation will be as provided in the following Theorem.
\begin{theorem}
The optimal power allocation that maximizes the achievable secrecy rate of the cooperative scenario with relaying is the solution of ~\eqref{5} subject to ~\eqref{5_sub} and ~\eqref{6} subject to ~\eqref{6_sub} identified with the optimal set $(P_{ab}^{*},P_{jb}^{*})$ in~\eqref{8} and~\eqref{9}  respectively with ${P_{ab}}^{*}= \alpha P_{jb}$
\end{theorem}

%

\section{COOPERATIVE MAC}
\label{COOPERATIVE MAC}
%

\subsection{Cooperative Scenario without Relaying}
%

This scenario will deal with the cooperative model in Figure \ref{fig:figure-1} without the need to relay anyone's data, so both will cooperate in their own transmissions power to maximize the secrecy rate of the other. Note that the achievable secrecy rate region is still the same, however, this is a special case of the one in the previous scenario, where the SNR that contributes to the extra rates through relayed data will disappear from the equation.

\vspace{0.3cm}

The achievable secrecy rates are defined as follows,
\begin{equation}
\label{10}
{Cs1}_{\alpha P_{j}}=\max_{\alpha P_{j}} \ log(1+{SNR}_{ab})-log(1+{SNR}_{ae}) 
\end{equation}
Subject to, 
\begin{equation}
\label{sub_10}
P_{j}\leq {P_{J}}
\end{equation}
\begin{equation}
\label{11}
{Cs2}_{\frac{P_{a}}{\alpha}}=\max_{\frac{P_{a}}{\alpha}} \ log(1+{SNR}_{jb})-log(1+S{NR}_{je}) 
\end{equation}
Subject to, 
\begin{equation}
\label{sub_11}
P_{a}\leq {P_{A}} 
\end{equation}

\vspace{0.3cm}

Solving~\eqref{10}, the cooperative optimal $P_{j}^{*}$ is the solution of the quadratic equation,
\footnotesize
\begin{equation}
\lambda \alpha^{2} G_{ab} G_{ae}P_{j}^{2}+\lambda(\alpha \sigma^{2} G_{ab}+ \alpha \sigma^{2} G_{ae})P_{j}-(\sigma^{2} G_{ab}-\sigma^{2} G_{ae}-\lambda \sigma^{4})=0 
\end{equation}
\normalsize

Solving~\eqref{11}, the cooperative optimal $P_{a}^{*}$ is the solution of the quadratic equation,
\small
\begin{equation}
\frac{\lambda G_{je} G_{jb}}{\alpha^{2}}P_{a}^{2}+\lambda (\frac{\sigma^{2}G_{jb}}{\alpha}+ \frac{\sigma^{2}G_{je}}{\alpha})P_{a}-(\sigma^{2}G_{jb}-\sigma^{2}G_{je}-\lambda \sigma^{4})=0 
\end{equation}     
\normalsize
Therefore, the optimal power allocation will be as provided in the following Theorem.

\vspace{0.3cm}

\begin{theorem}
The optimal power allocation that maximizes the achievable secrecy rates of the cooperative scenario without relaying is the solution of ~\eqref{10} subject to ~\eqref{sub_10} and ~\eqref{11} subject to ~\eqref{sub_11} identified with the optimal set $(P_{a}^{*},P_{j}^{*})$ given by the following closed forms,

\scriptsize 
\begin{eqnarray}
\label{40}
{P_{a}}^{*} = & \frac{\alpha}{2} \sqrt{\frac{{\lambda}^{2}{\sigma}^{4}(G_{jb}+G_{je})^{2}+4\lambda G_{jb}G_{je}({\sigma}^{2}G_{jb} 
-{\sigma}^{2}G_{je}-\lambda {\sigma}^{4})}{(\lambda G_{jb}G_{je})^{2}}}\nonumber \\
&-{\sigma}^{2}(\frac{1}{G_{je}}+\frac{1}{G_{jb}})&
\end{eqnarray}
\begin{eqnarray}
\label{41}
{P_{j}}^{*} = & \frac{1}{2\lambda\alpha} \sqrt{\frac{({\sigma}^{2}G_{ab}+{\sigma}^{2}G_{ae})^{2}+4\lambda G_{ab}G_{ae}({\sigma}^{2}G_{ab}-{\sigma}^{2}G_{ae}-\lambda {\sigma}^{4})}{(G_{ab}G_{ae})^{2}}}\nonumber  \\
&-{\sigma}^{2}(\frac{1}{G_{ae}}+\frac{1}{G_{ab}})&
\end{eqnarray}
\normalsize
\end{theorem}

Notice that the optimal cooperation level can be derived finding out $\frac{\partial P_{j}}{\partial \alpha}$  or $\frac{\partial P_{a}}{\partial \alpha}$. On the other hand, we can derive also the SNR over each link at which the cooperation is optimal.
%

\subsection{Cooperation from one side}
%

This scenario will consider that John helps Alice, while Alice does not help John. So, this scenario considers a cooperation at which Alice and John are concerned to maximize her own secrecy rate, since Eve is targeting Alice only. However, no relaying is considered here.

\vspace{0.3cm}

This scenario is a special case of Theorem 1, where the upper bound for the achievable secrecy rate will be as mentioned in the following Theorem.

\vspace{0.3cm}

\begin{theorem}
(Achievabilty): The achievable secure rate regions for the model in Figure 1 with cooperation from one side and without any relaying is upper bounded by the secure achievable rates for a MAC channel corresponding to each transmitter and their common receiver, and the channel between Alice and the eavesdropper. The rates in the closure of the union of all $(R_{1}, R_{2})$ satisfy,
\begin{equation}
R1\leq I(x_{a};Y_{b}|x_{j})-I(x_{a};Y_{e}).
\end{equation}
\begin{equation}
R2\leq I(x_{j};Y_{b}|x_{a}).
\end{equation}
\begin{equation}
R1+R2\leq I(x_{a},x_{j};Y_{b})-I(x_{a};Y_{e}).
\end{equation}
\end{theorem}

\begin{proof}
We can re-write the rate regions considering some auxiliary random variables $\acute{U}$, $\acute{V}$, and a time sharing variable $\acute{Q}$.
\begin{equation}
R1\leq I(x_{a};Y_{b}|x_{j},\acute{Q})-I(x_{a};Y_{e}|\acute{Q}). 
\end{equation}
\begin{equation}
R2\leq I(\acute{V},x_{j};Y_{b}|\acute{U},x_{a},\acute{Q}). 
\end{equation}
\begin{equation}
R1+R2\leq I(x_{a},x_{j},\acute{U},\acute{V};Y_{b}|\acute{Q})-I(x_{a},\acute{U};Y_{e}|\acute{Q}). 
\end{equation}

For some product distribution $\acute{U}\rightarrow \acute{V}\rightarrow (X_{a}, X_{j})\rightarrow y_{ab}, y_{jb}\rightarrow Y_{b}$, $y_{ae}\rightarrow Y_{e}$. The joint probability distributions $p(\acute{u},\acute{v},x_{a},x_{j},y_{ab},y_{jb})$ factors as, $p(\acute{u})p_(x_{a}|\acute{u})p(\acute{v})p(x_{j}|\acute{v})p(y_{ab}|x_{a})p(y_{jb}|x_{j})$ and the joint probability distribution $p(\acute{u},x_{a},y_{ae})$ factors as, $p(\acute{u}p_{xa}|\acute{u})p(y_{ae}|x_{a})$; and the time sharing variable $\acute{Q}$ and through Markovity, $\acute{U}\rightarrow \acute{V}\rightarrow (X_{a},X_{j})\rightarrow (Y_{b},Y_{e})$, the claim is proved.
\end{proof}

We define the optimization problem for the scenario of cooperation from one side as follows,
\begin{equation}
\label{48}
{Cs1}_{P_{a}}=\max_{P_{a}} \ log(1+{SNR}_{ab})-log(1+{SNR}_{ae})   
\end{equation}
Subject to, 
\begin{equation}
\label{sub_48}
P_{a}\leq {P_{A}}   
\end{equation}
\begin{equation}
\label{50}
{Cs2}_{\alpha P_{j}}=\max_{\alpha P_{j}} \ log(1+{SNR}_{ab})-log(1+{SNR}_{ae})
\end{equation}
Subject to, 
\begin{equation}
\label{sub_50}
P_{j}\leq {P_{J}}
\end{equation}

Such scenario will be a mixed scenario from the previous scenario and the one in the next section. The optimal non-cooperative $P_{a}^{*}$ is the solution of the quadratic equation,
\small
\begin{equation}
\lambda G_{ab} G_{ae}P_{a}^{2}+ \lambda(\sigma^{2} G_{ab}+\sigma^{2} G_{ae})P_{a}-(\sigma^{2} G_{ab}-\sigma^{2} G_{ae}-\lambda \sigma^{4})=0
\end{equation}
\normalsize

The optimal cooperative $P_{j}^{*}$ is the solution of the quadratic equation,
\footnotesize 
\begin{equation}
\lambda\alpha^{2}G_{ab} G_{ae}P_{j}^{2}+\lambda(\alpha \sigma^{2} G_{ab}+\alpha \sigma^{2}G_{ae})P_{j}-(\sigma^{2} G_{ab}- \sigma^{2}G_{ae}-\lambda \sigma^{4} )=0 
\end{equation}
\normalsize
Therefore, the optimal power allocation will be as provided in the following Theorem.

\vspace{0.3cm}

\begin{theorem}
The optimal power allocation that maximizes the achievable secrecy rates of the cooperative scenario from one side is the solution of ~\eqref{48} subject to ~\eqref{sub_48} and ~\eqref{50} subject to ~\eqref{sub_50} identified with the optimal set $(P_{a}^{*},P_{j}^{*})$ given by the following closed forms,

\scriptsize 
\begin{eqnarray}
\label{54}
{P_{a}}^{*} = &\frac{1}{2}\sqrt{\frac{{\lambda}^{2}{\sigma}^{4}(G_{ab}+G_{ae})^{2}+4\lambda G_{ab}G_{ae}({\sigma}^{2}G_{ab}-{\sigma}^{2}G_{ae}-\lambda {\sigma}^{4})}{(\lambda G_{ab}G_{ae})^{2}}} \nonumber \\
&-{\sigma}^{2}(\frac{1}{G_{ae}}+\frac{1}{G_{ab}})&
\end{eqnarray}

\begin{eqnarray}
\label{55}
{P_{j}}^{*} = & \frac{1}{2\lambda\alpha} \sqrt{\frac{({\sigma}^{2}G_{ab}+{\sigma}^{2}G_{ae})^{2}+4\lambda G_{ab}G_{ae}({\sigma}^{2}G_{ab}-{\sigma}^{2}G_{ae}-\lambda {\sigma}^{4})}{(G_{ab}G_{ae})^{2}}}\nonumber  \\
&-{\sigma}^{2}(\frac{1}{G_{ae}}+\frac{1}{G_{ab}})&
\end{eqnarray}
\normalsize

\end{theorem}

\vspace{0.03cm}

\section{NO-COOPERATION}
\label{NO-COOPERATION}

In this scenario, no cooperation exists, thus every device wants to maximize its own utility with its own resources. The reasoning behind considering this scenario is to study the implications of cooperation in the solution; i.e., to provide analytical insight when the cooperation is of benefit. 

\vspace{0.3cm}

We can first consider that the achievable secrecy rate regions for the model in Figure 1 with two non-cooperative devices and one eavesdropper coincide with that of the MAC with one common eavesdropper. Therefore, we can conclude the following Theorem. 

\vspace{0.3cm}

\begin{theorem}
\label{noncooperative}
(Achievability): The achievable secure rate regions for the model in Figure 1 is upper bounded by the "secure" achievable rates for a MAC channel corresponding to each non-cooperative transmitter and their common receiver, and the two legitimate non-cooperative transmitters and their common eavesdropper. The rates in the closure of the union of all $(R1,R2)$ satisfy,

\begin{equation}
R1\leq I(x_{a};Y_{b}|x_{j})-I(x_{a};Y_{e}|x_{j}). 
\end{equation}
\begin{equation}
R2\leq I(x_{j};Y_{b}|x_{a})-I(x_{j};Y_{e}|x_{a}). 
\end{equation}
\begin{equation}
R1+R2\leq I(x_{a},x_{j};Y_{b})-I(x_{a},x_{j};Y_{e}). 
\end{equation}

\end{theorem}

\vspace{0.3cm}

\begin{proof}
The proof of this Theorem follows the steps in the the proof of the Theorem 1.
\end{proof}

It is of particular relevance to notice that the achievable secrecy rates in Theorem 6 are only upper bounds to the rates when the data are decoded from each transmitter. In fact, the scenario of common receiver and common eavesdropper will provide the chance for one legitimate transmitter or the other to transmit securely and reliably or to leak data unevenly. To clarify this point, we can provide the following example. Let the achievable secrecy rate for Alice is as follows,

\begin{equation}
R1\leq I(x_{a};Y_{b}|x_{j})-I(x_{a};Y_{e}). 
\end{equation}
Then, based on the MAC channel rate regions defined by common legitimate and illegitimate receivers. The achievable rate for John will be as follows,
\begin{equation}
R2\leq I(x_{j};Y_{b})-I(x_{j};Y_{e}|x_{a}). 
\end{equation}
Based on the chain rule of the mutual information defined as,
\begin{equation}
I(x_{i},x_{j};Y_{j})=I(x_{i};Y_{j})+I(x_{j};Y_{j}|x_{i}). 
\end{equation}
The sum of the rate regions will be given as defined in the Theorem, i.e.,
\begin{equation}
R1+R2\leq I(x_{a},x_{j};Y_{b})-I(x_{a},x_{j};Y_{e}). 
\end{equation}

We define the optimization problem for the scenario of no cooperation as follows,
\begin{equation}
\label{14}
{Cs1}_{P_{a}}=\max_{P_{a}} \ log(1+{SNR}_{ab})-log(1+{SNR}_{ae})
\end{equation}
Subject to, 
\begin{equation}
\label{sub_14}
P_{a}\leq {P_{A}}
\end{equation}
\begin{equation}
\label{15}
{Cs2}_{P_{j}}=\max_{P_{j}} \ log(1+{SNR}_{jb})-log(1+{SNR}_{je})
\end{equation}
Subject to, 
\begin{equation}
\label{sub_15}
P_{j}\leq {P_{J}}
\end{equation}

\vspace{0.3cm}

Solving~\eqref{14}, the optimal non-cooperative $P_{a}^{*}$ is the solution of the quadratic equation,
\begin{equation}
\small
\lambda G_{ab} G_{ae}P_{a}^{2}+ \lambda(\sigma^{2} G_{ab}+\sigma^{2} G_{ae}) P_{a}-(\sigma^{2}G_{ab}-\sigma^{2}G_{ae}-\lambda \sigma^{4})=0
\end{equation}
\normalsize

Solving~\eqref{15}, the optimal non-cooperative $P_{j}^{*}$ is the solution of the quadratic equation,
\begin{equation}
\small
\lambda G_{jb}G_{je}P_{j}^{2}+ \lambda(\sigma^{2} G_{ab}+\sigma^{2}G_{je})P_{j}-(\sigma^{2}G_{jb}-\sigma^{2}G_{je}-\lambda \sigma^{4})=0
\end{equation}
\normalsize
%

Therefore, the optimal power allocation will be as provided in the following Theorem.

\begin{theorem}
The optimal power allocation that maximizes the achievable secrecy rates for the non-cooperative scenario is the solution of ~\eqref{14} subject to ~\eqref{sub_14} and ~\eqref{15} subject to ~\eqref{sub_15} identified with the optimal set $(P_{a}^{*},P_{j}^{*})$ given by the following closed forms,

\scriptsize 
\begin{eqnarray}
\label{69}
{P_{a}}^{*} = &\frac{1}{2}\sqrt{\frac{{\lambda}^{2}{\sigma}^{4}(G_{ab}+G_{ae})^{2}+4\lambda G_{ab}G_{ae}({\sigma}^{2}G_{ab}-{\sigma}^{2}G_{ae}-\lambda {\sigma}^{4})}{(\lambda G_{ab}G_{ae})^{2}}} \nonumber \\
&-{\sigma}^{2}(\frac{1}{G_{ae}}+\frac{1}{G_{ab}})&
\end{eqnarray}
\begin{eqnarray}
\label{70}
{P_{j}}^{*}= & \frac{1}{2}\sqrt{\frac{{\lambda}^{2}{\sigma}^{4}(G_{jb}+G_{je})^{2}+4\lambda G_{jb}G_{je}({\sigma}^{2}G_{jb}-{\sigma}^{2}G_{je}-\lambda {\sigma}^{4})}{(\lambda G_{jb}G_{je})^{2}}} \nonumber \\
&-{\sigma}^{2}(\frac{1}{G_{je}}+\frac{1}{G_{jb}})&
\end{eqnarray}
\normalsize
\end{theorem}

Notice that the solution set of this scenario is a special case of the solution set of the previous scenario when the cooperation level $\alpha=1$, as well as the constant $\lambda=1$. In fact, in a non-cooperative scenario, with per device total power constraint, it can be easily shown through the Lagrangian and the KKT conditions that the optimal strategy for each device is to allocate their own total maximum power, i.e., $(P_a^*,P_j^*)=(P_A,P_J)$.
%

\section{COOPERATIVE DISTANCE}
\label{COOPERATIVE DISTANCE}
%

Cooperation may be not beneficial for both parties, so no-cooperation will be one response from one or both devices if the cooperation will adversely affect its secrecy. Hence, it follows the importance of the distance considerations between cooperating parties. Thus, we will next consider the distance between Alice and John so that cooperation beneficially exists; otherwise John will adversely affect Alice. 
\begin{figure}[ht!]
    \begin{center}
      \mbox{\includegraphics[width=3.4in]{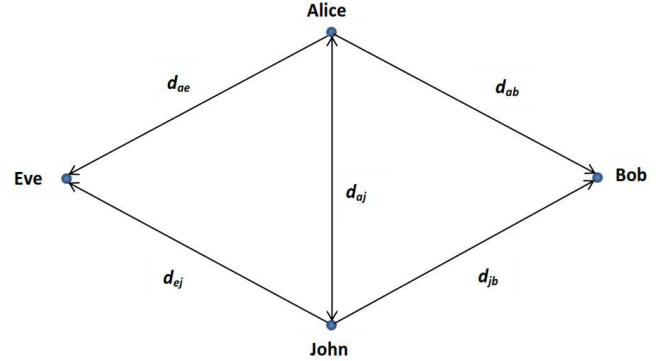}}
          \caption{System model illustrating fixed distances between the devices in the wireless cloud}
    \label{fig:figure-2}
    \end{center}
    \label{figure-2}
\end{figure}

\vspace{0.3cm}

Consider the model in Figure \ref{fig:figure-2}. The distances between different devices are considered such that $d_{ik}$ is the distance between transmitter $i$ and receiver $k$. Then, using the relation between the path loss exponent $d_{ik}^{-\eta}$ which relates the loss of the transmitted power over the distance of the transmission path, we can induce the distance into the SNR in the defined utilities. Consider $\eta=2$ and the cooperative scenario-A section VI, the optimal cooperative $P_{j}^{*}$ is the solution of the quadratic equation,
\begin{eqnarray}
& \lambda {\alpha}^{2}G_{ab}G_{ae}{P_{j}}^{2} +\lambda({\sigma}^{2}\alpha G_{ab}{d_{ab}}^{2}+{\sigma}^{2}\alpha G_{ae}{d_{ae}}^{2})P_{j} \nonumber \\
&-({\sigma}^{2}G_{ab}{d_{ae}}^{2}-{\sigma}^{2}G_{ae}{d_{ab}}^{2}-\lambda {\sigma}^{4}{d_{ab}}^{2}{d_{ae}}^{2})=0&
\end{eqnarray}

and the optimal cooperative ${P_{a}}^{*}$ is the solution of the quadratic equation,
\begin{eqnarray} 
& \frac{\lambda G_{je}G_{jb}}{{\alpha}^{2}}{P_{a}}^{2}+\lambda(\frac{{\sigma}^{2}G_{jb}{d_{je}}^{2}}{\alpha}+\frac{{\sigma}^{2}G_{je}{d_{jb}}^{2}}{\alpha})P_{a}\nonumber \\
&-({\sigma}^{2}G_{jb}{d_{jb}}^{2}-{\sigma}^{2}G_{je}{d_{je}}^{2}-\lambda {\sigma}^{4}{d_{jb}}^{2}{d_{je}}^{2})=0&
\end{eqnarray}

\vspace{0.3cm}

In fact, inducing the distances between Alice, Eve, and Bob, or John, Eve, and Bob is not enough to make cooperation exit. Therefore, we need to consider the distance between Alice and John in the cooperation problem. Consider the effect of Alice in John and vice versa as an interference effect, thus the cooperative optimization problem in scenario-A, Section-VI is such that the information rate from Alice to John and vice versa will influence one another in a positive way, i.e., it is used to cancel the rate decay (leakage) to Eve. So, in some way or another, the optimization problem can be written as follows, 
\begin{equation}
{Cs1}_{\alpha P_{j}}+log(1+{SNR}_{ja})
\end{equation}
\begin{equation}
{Cs2}_{P_{a}/\alpha}+log(1+{SNR}_{aj})
\end{equation}

\vspace{0.3cm}

From a power allocation perspective, this formulation means that, $log(1+SNR_{ja})$ would substitute for $log(1+SNR_{ae})$ and $log(1+SNR_{aj})$ would substitute for $log(1+SNR_{je})$ as well, otherwise cooperation will not exist. Substitute the distances into the conditions discussed we get the following,
\small
\begin{equation}
\frac{\alpha G_{aj}}{{d_{ab}}^{2}{\sigma}^{2}+\alpha G_{aj}P_{j}}\leq \frac{\alpha G_{ae}}{{d_{ae}}^{2}{\sigma}^{2}+\alpha G_{ae}P_{j}}
\end{equation}
\normalsize
and,
\small
\begin{equation}
\frac{G_{ja}}{\alpha {d_{ab}}^{2}{\sigma}^{2}+G_{ja}P_{a}}\leq \frac{G_{je}}{\alpha {d_{je}}^{2}{\sigma}^{2}+G_{je}P_{a}}
\end{equation}
\normalsize

This leads to the condition that, if we need the cooperation to be of benefit for one or both parties, then the following distance constraints should exist. Similar analysis of the effect of the distance between different devices in a wiretap setup without cooperative scenarios has concluded that there is a distance consideration for which the secrecy can exist, otherwise not, and they call it secrecy coverage distance,~\cite{BarRod2011}.

\vspace{0.3cm}

Now for the setup of cooperative devices for the model in Figure \ref{fig:figure-2}, the distance between Alice and Eve should be,
\small
\begin{equation}
{d_{ae}}^{\eta}\leq \frac{G_{ae}}{G_{aj}}{d_{aj}}^{\eta}
\end{equation}
\normalsize

and the distance between John and Eve should be, 
\small
\begin{equation}
{d_{je}}^{\eta}\leq \frac{G_{je}}{G_{ja}}{d_{aj}}^{\eta}
\end{equation}
\normalsize

If such distance consideration exists, then the cooperative power allocation strategies are optimal in the sense of optimal cooperation level.

\vspace{0.3cm}

Hence, Eve can try to break such distance constraints going more near to one or both devices she wants to eavesdrop, i.e., moving the cooperation level into less cooperative and so the achievable secure and reliable rates into lower bounds. 
%

\section{ALGORITHM}
\label{ALGORITHM}
%

We introduce a distributed algorithm that finds the optimal power allocation set to secure the data transmission for the model in Figure \ref{fig:figure-1}. First, the devices will check the distance constraints to test if cooperation is of benefit. If yes, then the devices will initiate the cooperation and will decides to cooperate jointly, not to cooperate, or to cooperate from one side. Therefore, the optimal power allocation for Alice and John will follow the solution set of the scenarios discussed.
\begin{algorithm}[h!]
%
Alice$\rightarrow$initiates cooperation mode.\\
Alice$\rightarrow$requests relay service
\BlankLine
\SetKwInOut{Input}{Input}\SetKwInOut{Output}{Output}
			\DontPrintSemicolon
            \BlankLine			
            \Input{ distance $d_{ab}$, $d_{ae}$, $d_{jb}$, $d_{je}$, $d_{aj}$.}
			\BlankLine
			%
			\uIf{{$\frac{\alpha G_{aj}}{{d_{ab}}^{2}{\sigma}^{2}+\alpha G_{aj}P_{j}}\leq \frac{\alpha G_{ae}}{{d_{ae}}^{2}{\sigma}^{2}+\alpha G_{ae}P_{j}}$} and 
			 \BlankLine
			{$\frac{G_{ja}}{\alpha {d_{ab}}^{2}{\sigma}^{2}+G_{ja}P_{a}}\leq \frac{G_{je}}{\alpha {d_{je}}^{2}{\sigma}^{2}+G_{je}P_{a}}$} and 
			 \BlankLine
			 {${d_{ae}}^{2}\leq \frac{G_{ae}}{G_{aj}}{d_{aj}}^{2}$} and 
			  \BlankLine
			 {${d_{je}}^{2}\leq \frac{G_{je}}{G_{ja}}{d_{aj}}^{2}$}}
			 \BlankLine
			{
				\Indp
								John$\rightarrow$accepts to cooperate and decide cooperation with level $\alpha$ and request Alice relay service. 
				}
				 \BlankLine
				\uIf{Alice accepts to cooperate;}{${Output:\bf{1}}$ is executed.}

			\vspace{0.01cm}
			
			\uElseIf{John$\rightarrow$rejects to relay data and devices go to MAC cooperative mode;}{${Output:\bf{2}}$ is executed.}
			
			\vspace{0.01cm}
			
			\uElseIf{John$\rightarrow$accepts to cooperate from one side;}{${Output:\bf{3}}$ is executed.}
			
			\vspace{0.01cm}
			
			\uElseIf{John$\rightarrow$rejects to cooperate and devices go to non-cooperative mode;}{${Output:\bf{4}}$ is executed.}
			
			\vspace{0.01cm}
			
			\uElse{${Output:\bf{4}}$ is executed.}
			 \BlankLine
			\Output{\bf{1}}{
              The optimal cooperative relay power:\\
              \vspace{0.1cm}
              ${P_{ab}}^{*}$ solving equation~\eqref{8}\\
              \vspace{0.1cm}
              ${P_{jb}}^{*}$ solving equation~\eqref{9}
}
        \BlankLine
       \vspace{0.1cm}
        \Output{\bf{2}}{
              The optimal cooperative main power:\\  
               ${P_{a}}^{*}$ in~\eqref{40}\\
              \vspace{0.1cm}
              ${P_{j}}^{*}$ in~\eqref{41}
}
        \BlankLine
        \vspace{0.1cm}
        \Output{\bf{3}}{

              The optimal cooperative / non-cooperative main power:\\ 
             \vspace{0.1cm}
             ${P_{a}}^{*}$ in~\eqref{54}\\
              \vspace{0.1cm}
              ${P_{j}}^{*}$ in~\eqref{55}                          
%
%
}
              \BlankLine
              \vspace{0.1cm}
        \Output{\bf{4}}{
              The optimal non-cooperative main power:\\ 
              ${P_{a}}^{*}$ in~\eqref{69}\\
              \vspace{0.1cm}
              ${P_{j}}^{*}$ in~\eqref{70}
}
              \BlankLine              
Devices keep checking distance constraint and adaptively allocate their optimal power based on the cooperation scenario selected.
\caption{Optimum Cooperative Power Allocation.}\label{algo}
\end{algorithm}

\section{Numerical Results}
\label{results}

We shall now present a set of illustrative results that cast further insights to the problem. We choose a cooperation level $\alpha=0.8,$ channels gains are $G_{ab}=0.4, \ G_{ae}=0.3, \ G_{ja}=0.2, \ G_{jb}=0.5,$ and $G_{je}=0.3$. We now analyze the set of scenarios considered. Notice that we have chosen channel gains for a non-degraded case, where the channels between legitimate transmitters and the legitimate receiver are stronger than those between legitimate transmitters and the illegitimate receiver. Figure \ref{fig:figure-3} illustrates the achievable secrecy rate for Alice with respect to the main power $P_a$ and the power $P_{jb}$ used to relay her. Figure \ref{fig:figure-4} illustrates the achievable secrecy rate for John with respect to the main power $P_j$ and the power $P_{ab}$ used to relay his data. As expected, the framework of cooperation via relaying adds significantly to the achievable secure data rates of each device compared to the data rates achieved without cooperation. The difference between the gains in the data rates for Alice and John is due to the stronger channel gain that John enjoys between his device and the receiver device, Bob.

\vspace{0.3cm}

Figure \ref{fig:figure-5} illustrates the achievable secrecy rates of Alice with respect to her distance from Bob. The secrecy rates has been simulated under different distances between Alice and Eve $d_{ae}$ and between John and Bob $d_{jb}$. As already explained analytically in the previous sections, such distances are associated to the SNRs obtained without and with relaying. Therefore, it is of particular relevance to observe that the distance of the eavesdropper and the transmitter has fundamental role in deciding whether the cooperation is of benefit or not. This result shows that as long as the distances between the legitimate transmitters and the legitimate receiver are smaller than the distances between the legitimate transmitters and the eavesdropper, the achievable secrecy gains are noticeable, and as expected with relaying, the secrecy rates are higher. Therefore, cooperative relaying is of benifit. However, its interestingly shown that if John is too much far from Bob, the gains expected from relaying are very limited, and at some point going into no cooperation could be of more benefit to the legitimate transmitter. Similar analysis applies to the achievable secrecy rates of John with respect to his distance from Bob.
\begin{figure}[ht!]
    \begin{center}
        \mbox{\includegraphics[height=2in]{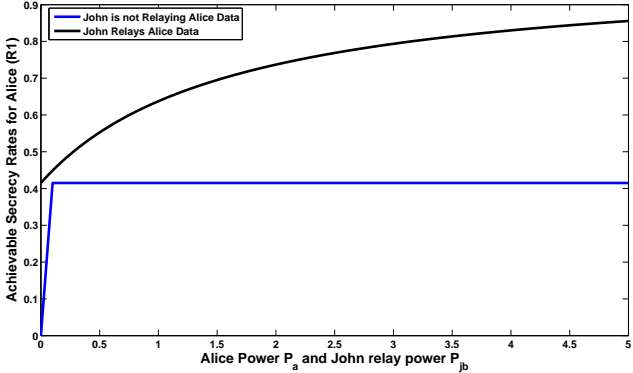}}
    \caption{Secure data rates achievable by Alice (with and without relaying) with respect to the main power $P_a$, (when no cooperation exists) and the relay power $P_{jb}$ (when $P_a=P_A$ is fixed and equals 5 and cooperative relaying is active).}
    \label{fig:figure-3}
    \end{center}
    \label{figure-3}
\end{figure}
\begin{figure}[ht!]
    \begin{center}
        \mbox{\includegraphics[height=2in]{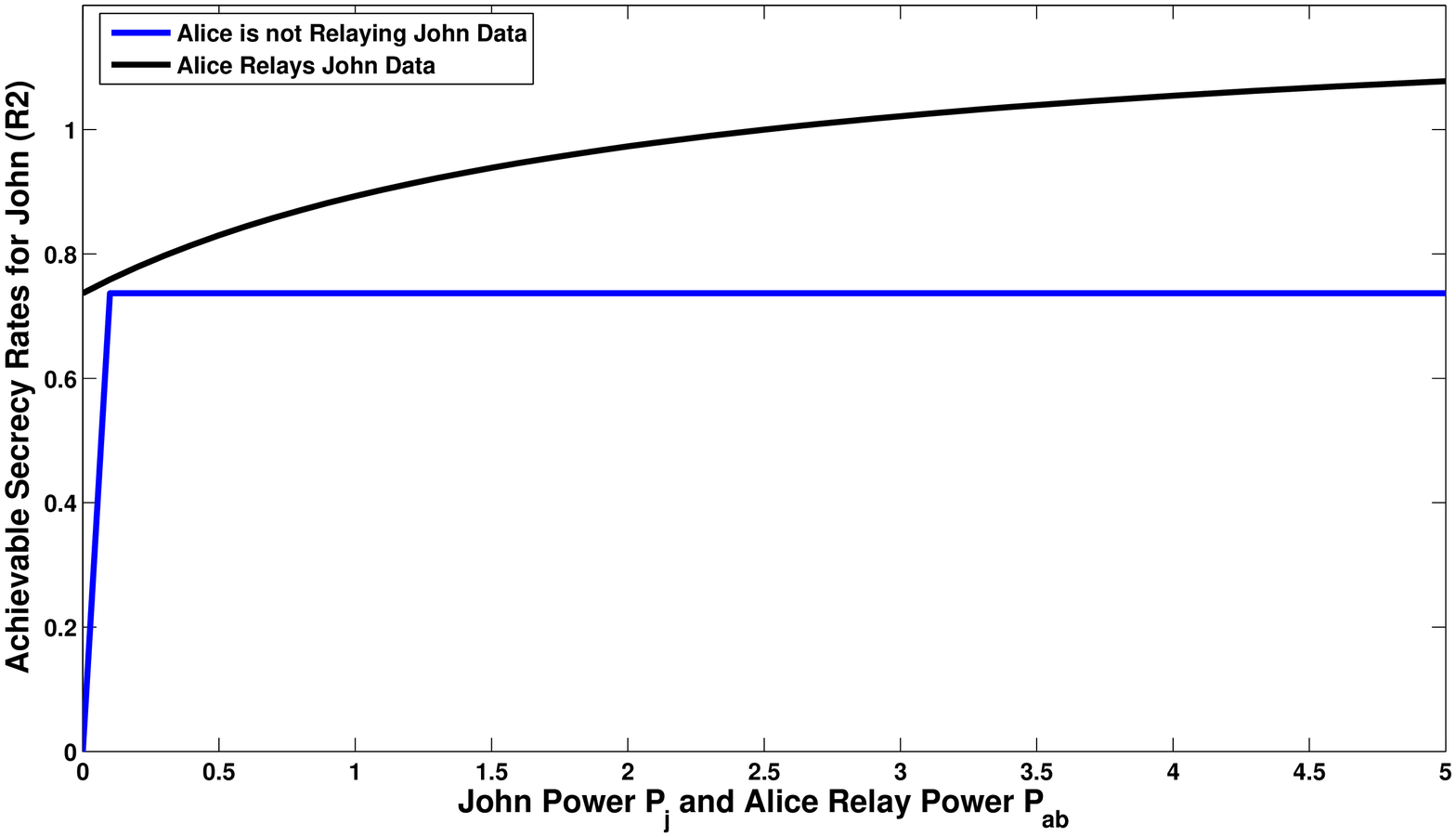}}
    \caption{Secure data rates achievable by John (with and without relaying) with respect to the main power $P_j$, (when no cooperation exists) and the relay power $P_{ab}$ (when $P_j=P_J$ is fixed and equals 5 and cooperative relaying is active).}
    \label{fig:figure-4}
    \end{center}
    \label{figure-4}
\end{figure}

\vspace{0.3cm}

\begin{figure}[ht!]
    \begin{center}
        \mbox{\includegraphics[height=2in]{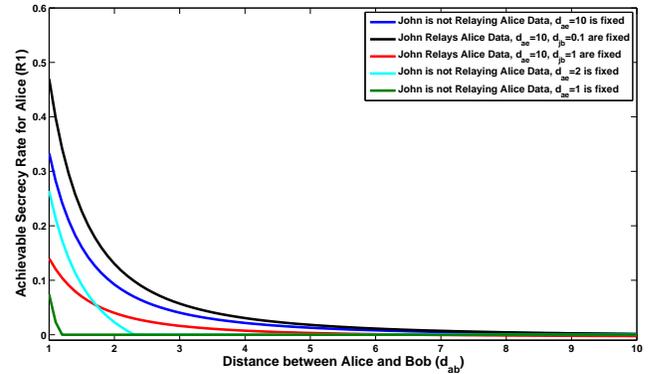}}
    \caption{Secure data rates achievable by Alice, with and without relaying, with respect to the distance between Alice and Bob $(d_{ab})$, and under different distances $(d_{ae}, d_{jb}$ between Alice and Eve and John and Bob, respectively.}
    \label{fig:figure-5}
    \end{center}
    \label{figure-5}
\end{figure}
%
The results in Figure \ref{fig:figure-6} to Figure \ref{fig:figure-8} show the achievable secrecy rate with respect to the optimized cooperative power of the other device. Note that any set of optimum points will lead to a point at which the achievable secrecy rate is of maximum value. Note also that, it is clear that the optimal power allocation in the scenario-VII will be to use the maximum power so the achievable rates will be fixed for all fixed instantaneous measures of the channels gains and thus we didn't include it here. In Figure \ref{fig:figure-6} which considers cooperative scenario-A with relaying, we can see the non-linear behavior in terms of the optimal cooperative power. We can also see that higher rates are achievable at the same cooperative powers in comparison to scenario-A in Figure \ref{fig:figure-7} where relaying is not implemented. 

\vspace{0.3cm}

Therefore, with power cooperation and without data being relayed, the achievable rates are supposed to be less. We can also see at one chosen point of cooperative relaying power for Alice and John; i.e., for a specific achievable rate  of 20, the corresponding cooperative power set $(P_{ab},P_{jb})$ is (8.32, 11.25) for Alice and John, respectively. We can choose any other optimal set for scenario-A with relaying in Figure \ref{fig:figure-6}. Similarly, for scenario-A  in Figure \ref{fig:figure-7} without relaying, we see that at the cooperative power set $(P_{a},P_{j})$ equals (7.76, 13) the achievable secrecy rate is 10 via the same process. We can also see that both Figure \ref{fig:figure-6} and Figure \ref{fig:figure-7} show the effect on the secrecy rates when the distance conditions are met, we can see that the cooperation is of benefit for both devices, Alice and John, and their cooperation leads to maximizing their achievable secrecy rates.
\begin{figure}[ht!]
    \begin{center}
        \mbox{\includegraphics[height=2in]{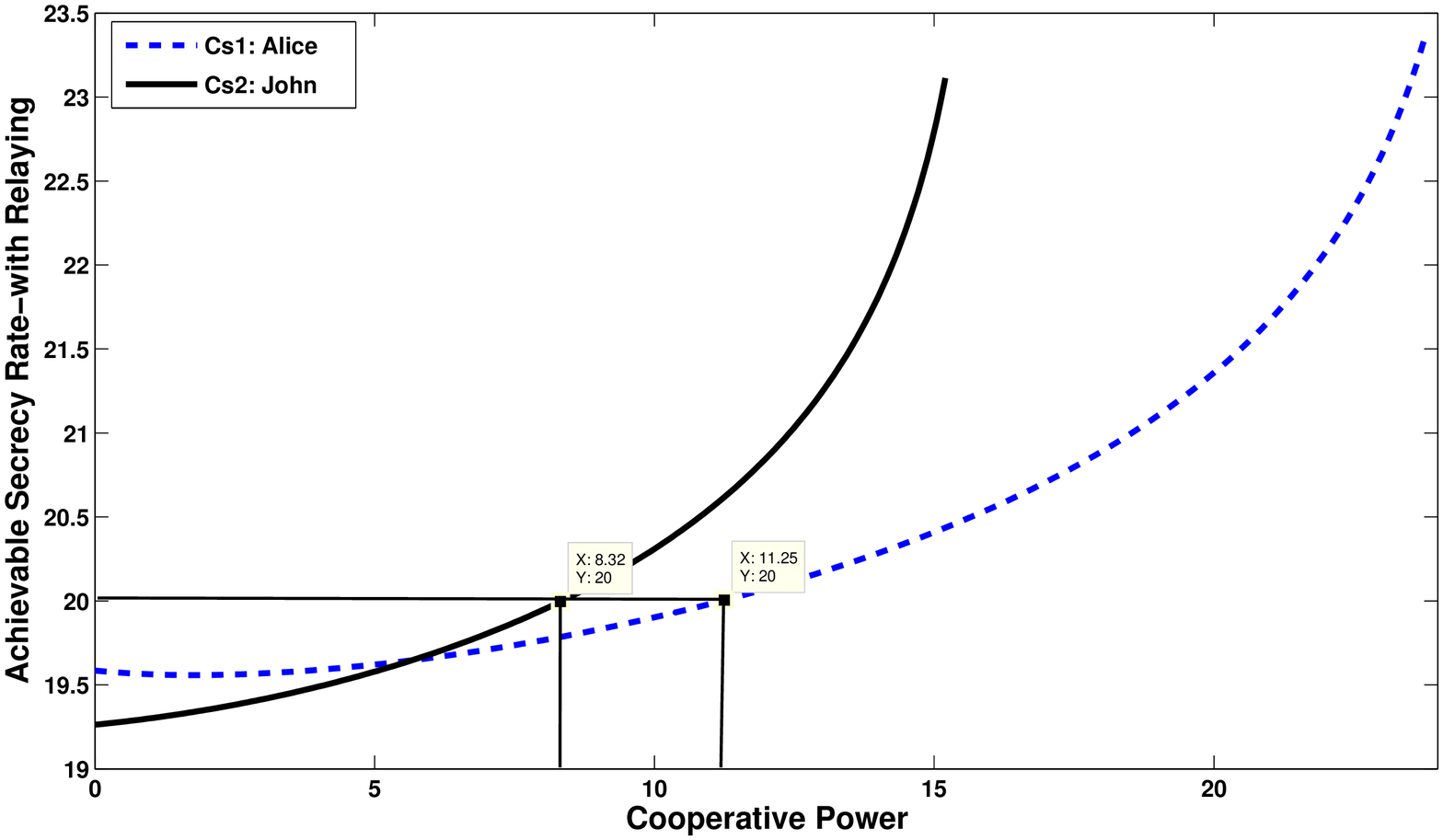}}
    \caption{Achievable secrecy rates, with relaying and with distance constraints met, versus the optimal cooperative power $({P_{ab}}^{*}, {P_{jb}}^{*})$.}
    \label{fig:figure-6}
    \end{center}
    \label{figure-6}
\end{figure}
\begin{figure}[ht!]
    \begin{center}
        \mbox{\includegraphics[height=2in]{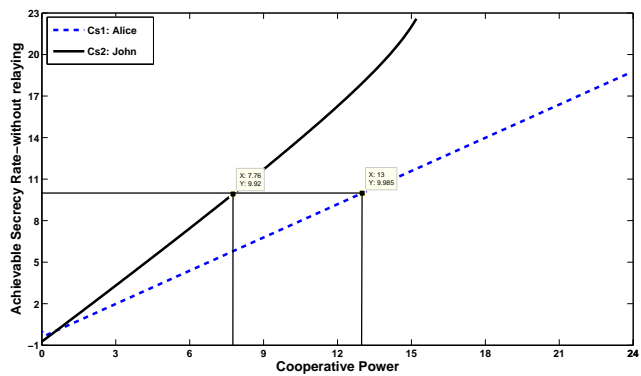}}
    \caption{Achievable secrecy rates, without relaying and with distance constraints met, versus the optimal cooperative power $({P_{a}}^{*}, {P_{j}}^{*})$.}
    \label{fig:figure-7}
    \end{center}
    \label{figure-7}
\end{figure}

\vspace{0.3cm}

However, we can see in Figure \ref{fig:figure-8}, where the distance condition is not met, thus cooperation will be of no benefit at least for one of the devices, i.e., the achievable secrecy rate for the device that doesn't obey the distance constraints decreases as the cooperation increases from the other device. Therefore, as shown in Figure \ref{fig:figure-8}, the cooperation is adversely affecting the secrecy rate for Alice; since the condition of the distance between Eve and Alice compared to the distance between John and Alice is not met, here we can easily pick up a point where devices will chose that cooperation will not exist, that is the point of intersection, which can also correspond to equal cooperation from both sides, that is at $\alpha=1.$ 
\begin{figure}[ht!]
    \begin{center}
        \mbox{\includegraphics[height=2in]{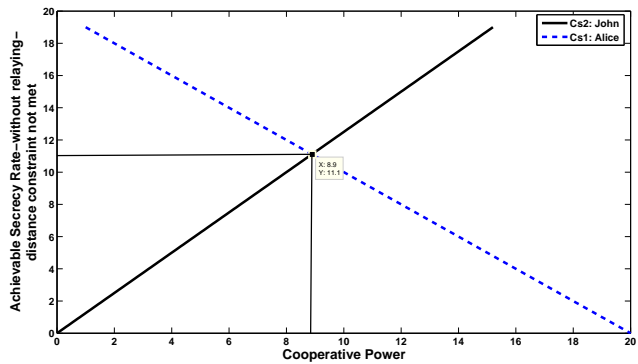}}
    \caption{Achievable secrecy rates, without relaying and with distance constraints not met, versus the optimal cooperative power $({P_{a}}^{*}, {P_{j}}^{*})$.}
    \label{fig:figure-8}
    \end{center}
    \label{figure-8}
\end{figure} 
%

\section{Conclusion}
\label{conclusions}
%

The main contribution of this paper is a classification of the achievable secrecy rate regions for a model with cooperative relaying at the two legitimate transmitters and with no relaying where the model forms a MAC channel. A derivation of different optimal cooperative power allocation interpretations that aim to maximize the achievable secrecy rates for the devices cooperating in a wireless medium have been derived. We exploited the cooperation concept emulating the wireless medium fading via the cooperative diversity induced by relaying; that is known to add positively to the achievable rates and consequently to the achievable secrecy rates. We evaluated the cooperation schemes to the non-cooperative ones in a distributed optimization formulation. We propose a distributed algorithm to measure the optimal power allocation. In particular, the optimal cooperative power allocation in the different scenarios interestingly follows an inverse waterfilling interpretation. We build upon distance constraints that define when the cooperation will be of benefit or not; i.e., the distance constraints will be the trigger for the cooperative service request which leads to an optimal power control in an adaptive sense, specifically if mobility is considered and channel state information is shared between the legitimate transmitting devices. This setup can contribute to some clustering criteria in certain wireless setups, through which such constraints are induced into the main optimal cooperative power allocation. The paper sheds light on the feasibility of cooperation in wireless clouds. The cooperation concept seems so promising and appealing for next generation wireless networks, however, to secure data transmission such device to device cooperation should be associated with adaptive and distributed algorithms that constraint the global cloud cooperation when cooperating devices will cause interference and so jam the main transmission or when possible common or active eavesdroppers  exist. From a security perspective, there is should be a framework for cooperation that can be customizable to the application.

\bibliographystyle{IEEEtran}
\bibliography{IEEEabrv,mybibfile}

\end{document}